\documentclass[letterpaper,11pt]{article}

\usepackage{amsmath}
\usepackage{amsfonts}
\usepackage{amssymb}
\usepackage{amsthm}
\usepackage[usenames]{color}
\usepackage{fullpage}
\usepackage{xspace}
\usepackage{url,ifthen}
\usepackage{srcltx}
\usepackage{multirow}
\usepackage{boxedminipage}
\usepackage[margin=1.2in]{geometry}
\usepackage{nicefrac}
\usepackage{xspace}
\usepackage{graphicx}
\usepackage[usenames]{color}
\usepackage{srcltx}
\definecolor{DarkGreen}{rgb}{0.1,0.5,0.1}
\definecolor{DarkRed}{rgb}{0.5,0.1,0.1}
\definecolor{DarkBlue}{rgb}{0.1,0.1,0.5}
\usepackage[small]{caption}

\usepackage[pdftex]{hyperref}
\hypersetup{
    unicode=false,          
    pdftoolbar=true,        
    pdfmenubar=true,        
    pdffitwindow=false,      
    pdfnewwindow=true,      
    colorlinks=true,       
    linkcolor=DarkBlue,          
    citecolor=DarkGreen,        
    filecolor=DarkRed,      
    urlcolor=DarkBlue,          
    pdftitle={Beyond worst-case analysis in private singular vector computation},
    pdfauthor={Moritz Hardt, Aaron Roth},
}

\def\draft{0}

\def\submit{1}

\ifnum\draft=1 
    \def\ShowAuthNotes{1}
\else
    \def\ShowAuthNotes{0}
\fi

\ifnum\submit=1
\newcommand{\forsubmit}[1]{#1}
\newcommand{\forreals}[1]{}
\else
\newcommand{\forreals}[1]{#1}
\newcommand{\forsubmit}[1]{}
\fi

\ifnum\ShowAuthNotes=1
\newcommand{\authnote}[2]{{ \footnotesize \bf{{\color{DarkRed}[#1's Note:}
{\color{DarkBlue}#2}]}}}
\else
\newcommand{\authnote}[2]{}
\fi

\newtheorem{theorem}{Theorem}[section]
\newtheorem{remark}[theorem]{Remark}
\newtheorem{lemma}[theorem]{Lemma}
\newtheorem{corollary}[theorem]{Corollary}

\newtheorem{claim}[theorem]{Claim}
\newtheorem{fact}[theorem]{Fact}

\theoremstyle{definition}
\newtheorem{definition}[theorem]{Definition}

\newcommand{\chapterref}[1]{\hyperref[ch:#1]{Chapter~\ref{ch:#1}}}

\newcommand{\claimref}[1]{\hyperref[claim:#1]{Claim~\ref{claim:#1}}}
\newcommand{\corollarylabel}[1]{\label{cor:#1}}
\newcommand{\corollaryref}[1]{\hyperref[cor:#1]{Corollary~\ref{cor:#1}}}
\newcommand{\definitionlabel}[1]{\label{def:#1}}
\newcommand{\definitionref}[1]{\hyperref[def:#1]{Definition~\ref{def:#1}}}
\newcommand{\equationlabel}[1]{\label{eq:#1}}
\newcommand{\equationref}[1]{\hyperref[eq:#1]{Equation~\ref{eq:#1}}}
\newcommand{\factlabel}[1]{\label{fact:#1}}
\newcommand{\factref}[1]{\hyperref[fact:#1]{Fact~\ref{fact:#1}}}
\newcommand{\figurelabel}[1]{\label{fig:#1}}
\newcommand{\figureref}[1]{\hyperref[fig:#1]{Figure~\ref{fig:#1}}}
\newcommand{\itemlabel}[1]{\label{item:#1}}
\newcommand{\itemref}[1]{\hyperref[item:#1]{Item~\ref{item:#1}}}
\newcommand{\lemmalabel}[1]{\label{lem:#1}}
\newcommand{\lemmaref}[1]{\hyperref[lem:#1]{Lemma~\ref{lem:#1}}}

\newcommand{\propref}[1]{\hyperref[prop:#1]{Proposition~\ref{prop:#1}}}

\newcommand{\propositionref}[1]{\hyperref[prop:#1]{Proposition~\ref{prop:#1}}}

\newcommand{\remarkref}[1]{\hyperref[rem:#1]{Remark~\ref{rem:#1}}}
\newcommand{\sectionlabel}[1]{\label{sec:#1}}
\newcommand{\sectionref}[1]{\hyperref[sec:#1]{Section~\ref{sec:#1}}}
\newcommand{\theoremlabel}[1]{\label{thm:#1}}
\newcommand{\theoremref}[1]{\hyperref[thm:#1]{Theorem~\ref{thm:#1}}}

\usepackage[T1]{fontenc}
\usepackage{kpfonts}
\usepackage{microtype}

\newcommand{\Esymb}{\mathbb{E}}
\newcommand{\Psymb}{\mathbb{P}}
\newcommand{\Vsymb}{\mathbb{V}}
\DeclareMathOperator*{\E}{\Esymb}
\DeclareMathOperator*{\Var}{\Vsymb}
\DeclareMathOperator*{\ProbOp}{\Psymb}
\renewcommand{\Pr}{\ProbOp}

\newcommand{\mper}{\,.}
\newcommand{\mcom}{\,,}

\renewcommand{\hat}{\widehat}

\newcommand{\defeq}{\stackrel{\small \mathrm{def}}{=}}
\renewcommand{\leq}{\leqslant}
\renewcommand{\le}{\leqslant}
\renewcommand{\geq}{\geqslant}
\renewcommand{\ge}{\geqslant}

\newcommand{\Set}[1]{\left\{#1\right\}}

\newcommand{\signs}{\{-1,1\}}

\newcommand{\R}{\mathbb{R}}

\usepackage{bm}

\newcommand{\sign}{\mathit{sign}}

\renewcommand{\epsilon}{\varepsilon}

\newcommand{\remove}[1]{}

\newcommand{\ignore}[1]{}

\title{Beyond worst-case analysis in\\ private singular vector computation}
\author{Moritz Hardt\thanks{IBM Almaden Research. Email: {\tt
mhardt@us.ibm.com}} \and Aaron Roth\thanks{Department of Computer and Information Sciences, University of Pennsylvania. Supported in party by NSF grant CNS-1065060. Email: {\tt aaroth@cis.upenn.edu}}}

\begin{document}
\maketitle
\begin{abstract}
We consider differentially private approximate singular vector computation.
Known worst-case lower bounds show that the error of any differentially
private algorithm must scale polynomially with the dimension of the singular
vector. We are able to replace this dependence on the dimension by a natural
parameter known as the \emph{coherence} of the matrix that is often
observed to be significantly smaller than the dimension
both theoretically and empirically.
We also prove a matching lower bound showing that our guarantee is
nearly optimal for every setting of the coherence parameter.
Notably, we achieve our bounds by giving a robust analysis of the well-known
power iteration algorithm, which may be of independent interest. Our algorithm also leads to improvements in
worst-case settings and to better low-rank approximations in the spectral norm.
\end{abstract}
\vfill
\thispagestyle{empty}
\pagebreak
%
%
%
%

\section{Introduction}
Spectral analysis of graphs and matrices is one of the most fundamental tools in
data mining. The singular vectors of data matrices are used for spectral
clustering, principal component analysis, latent semantic indexing, manifold
learning, multi-dimensional scaling, low rank matrix approximation,
collaborative filtering, and matrix completion.
They provide a means of avoiding the curse of dimensionality by
discovering an (approximate) low-dimensional representation of seemingly very
high dimensional data. Unfortunately, many of the datasets for which
spectral methods are ideal are composed of sensitive user information:
browsing histories, friendship networks, movie reviews, and other data
collected from private user interactions. The Netflix prize dataset is a
perfect example of this phenomenon: a dataset of supposedly "anonymized" user
records was released for the Netflix Prize Challenge, which was a matrix of
user/movie review pairs. The goal of the competition was to
predict user/movie review pairs missing from the matrix.  Unfortunately, the
ad-hoc anonymization of this dataset proved to be insufficient, and Narayanan
and Shmatikov \cite{NarayananS08} were able to re-identify many of the users. Because
of the privacy concerns that the attack brought to light, the
second proposed Netflix challenge was canceled.

In the past decade, a rigorous formulation of privacy known as differential
privacy has been developed, along with a collection of powerful theoretical
results.  With very few exceptions, existing algorithms come with utility guarantees that hold in the worst case over the choice of the private data. As a result, these utility bounds can sometimes be too weak to be meaningful on particular data sets of interest.

Several algorithms are known for computing approximate top singular vectors of a matrix under differential privacy. In fact, nearly optimal error
bounds are known in the worst case. Unfortunately, differential privacy
unavoidably forces these bounds to degrade with the
dimension of the data. More concretely, given an $n\times n$ matrix $A,$ any
differentially private algorithm must in the worst case
output a vector $x$ such that $\|Ax\|_2\le \sigma_1(A) - O(\sqrt{n}),$ where
$\sigma_1(A)$ denotes the top singular value of $A.$ If the matrix $A$ has
bounded entries and is sparse as is very common, the dependence on $n$ in the
error term can easily overwhelm the signal.
This dependence on $n$ is discouraging, because one of the most compelling goals
of tools such as PCA is to overcome the ``curse of dimensionality'' inherent
in the analysis of very high dimensional data.  We therefore ask the question:
Can we hope to achieve a nearly \emph{dimension-free} bound under a reasonable
assumption on the input matrix?

We answer this question in the affirmative. Specifically, we give
an algorithm to compute an approximate singular vector that achieves
error $O(\sqrt{\mu(A)}\log(n)).$ Here, $\mu(A)$ denotes the \emph{coherence}
of the input matrix. The coherence varies between $1$ and $n.$ We say that $A$
has low coherence if $\mu(A)$ is significantly smaller than $n.$
Roughly, a matrix has low coherence if none of its singular vectors have any
large coordinates.  Low coherence is a widely observed property of large
matrices.  Random models exhibit low coherence as well as many real-world
matrices.  Indeed, many recent results in matrix completion, Robust Principal
Component Analysis and Low-rank approximation rely crucially on the assumption
that the input matrix has low coherence.
%
%
The error of our algorithm depends essentially only on the square root of the
coherence of the data matrix.
Moreover, we show that the exact dependence on the coherence that
we achieve is best possible: Specifically,
for each value of the coherence parameter, we give a family of matrices for
which no differentially private algorithm can get a better approximation to
the top singular vector than our algorithm does, up to logarithmic factors.


Our algorithm is also highly efficient and can be implemented using a nearly
linear number of vector inner product computations. In particular, our running
time is nearly linear in the number of nonzeros of the matrix. In fact, our algorithm is
a new variant of the classical power iteration method that has long been the
basis of many practical eigenvalue solvers.

\subsection{Our Results}
We say that a matrix $A \in \mathbb{R}^{m\times n}$ with singular value
decomposition $A = U\Sigma V^t$ has \emph{coherence}
\[
\mu(A) \defeq\left\{m\|U\|^2_\infty, n\|V\|^2_\infty\right\}\mper
\]
For now we assume that $m=n,$ but all of our results apply to general
matrices. Note that $\mu(A) \in [1,n]$.
We give a simple $(\epsilon,\delta)$-differentially
private algorithm which achieves the following guarantee.

\begin{theorem}[Informal, some parameters hidden]
For any matrix $A$ that satisfies a mild assumption on the decay of its singular values, Private Power Iteration returns a vector~$x$ such that with high probability
\[
\frac{\|Ax\|}{\|x\|} \geq \sigma_1(A) - O\left(
\epsilon^{-1}\sqrt{\mu(A)\log(1/\delta)}\log n\right)
\]
\end{theorem}

We also show a nearly matching lower bound:
\begin{theorem}[Informal]
For any coherence parameter $c \in \{2,\dots,n\}$, there exists a family of matrices
$\mathcal{A}$ such that for each $A \in \mathcal{A}$, $\mu(A) = c$, and such
that for every $(\epsilon,\delta)$-differentially private algorithm $M$ with
$\delta = \Omega(1/n)$ there is a matrix $A \in \mathcal{A}$ so that with high
probability, $M(A)$ outputs a vector $x$ such that
\[
\frac{\|A x\|_2}{\|x\|_2} \leq \sigma_1(A) -
\Omega\left(\epsilon^{-1}\sqrt{\mu(A)}\right)
\]
\end{theorem}
Note that in addition to showing that our dependence on $\mu(A)$ is tight,
this theorem shows that the error of any data-independent guarantee must be at
least $\Omega\left(\epsilon^{-1}\sqrt{n}\right)$.


Finally, we show how our algorithm can be used to compute accurate rank
$k$-approximations to the private matrix $A$ in the spectral norm, for any
$k$. For $k = 1$, the quality of our approximation is optimal. For $k \geq 2$,
as in previous work \cite{HardtR12}, our bounds depend on $r$, where $r$ is
the rank of $A$. Note that these bounds still improve on the best worst-case
bounds when $A$ is low rank.

\begin{theorem}[Informal, some parameters hidden]
There is an $(\epsilon,\delta)$-differentially private algorithm such that
for any matrix $A$ that satisfies a mild assumption on the decay of its singular values,
it  returns a rank-1 matrix $A_1$
such that with high probability
\[
\|A-A_1\|_2 \leq \sigma_2(A) + O\left(\epsilon^{-1}\sqrt{\mu(A)\log(1/\delta)}\log n\right)
\]
Moreover, there is an $(\epsilon,\delta)$-differentially private algorithm
such that for any rank $r$ matrix $A$ that satisfies a mild assumption on the
decay of its singular values, it  returns a rank-$k$ matrix $A_k$ such that with
high probability: $$\|A-A_k\|_2 \leq \sigma_{k+1}(A) +
O\left(\epsilon^{-1}k^2\sqrt{(r\cdot \mu(A)+k\log n)\log(1/\delta)}\log
n\right)$$
\end{theorem}

\subsection{More efficient and improved worst-case bounds}

Our robust power iteration analysis can also be applied easily to worst-case
settings without any incoherence assumptions. For example, we resolve multiple
questions asked by Kapralov and Talwar~\cite{KapralovT13}. Specifically, we
improve the running time of their algorithm by large polynomial factors, give a
much simpler algorithm and improve the error dependence on~$k.$ In the main body of the
paper we study differential privacy under changes of single entries. Here, we
consider unit changes in spectral norm as proposed by~\cite{KapralovT13}. Our
algorithm easily adapts to this definition and gives the following corollary.
\begin{corollary}
There is an algorithm
such that for every matrix $A$ that satisfies a mild assumption on the
decay of its singular values, it returns a rank-$k$ matrix $A_k$ such that with
high probability,
\begin{equation}\equationlabel{spectral}
\|A-A_k\|_2 \leq \sigma_{k+1}(A) +
O\left(\epsilon^{-1}k^2\sqrt{n\log(1/\delta)}\log n\right)\mper
\end{equation}
Moreover, the algorithm satisfies $(\epsilon,\delta)$-differential privacy
under unit spectral perturbations. For $(\epsilon,0)$-differential privacy the
error bound satisfies
\[
\|A-A_k\|_2 \leq \sigma_{k+1}(A) +
O\left(\epsilon^{-1}k^2n\log n\right)\mper
\]
\end{corollary}
We stress that \equationref{spectral} is the first bound for
$(\epsilon,\delta)$-differential privacy under unit spectral norm
perturbations. The dependence on $n$ matches the error achieved by randomized
response for single entry changes.
%
%
%
\subsection{Our Techniques}
Our main technical contribution includes a novel ``robust'' analysis of the
classical power iteration algorithm for computing the top eigenvector of a
matrix, which may be of independent interest. Specifically, we analyze power
iteration in which an arbitrary sequence of perturbations $g_1,\ldots,g_t$ may
be added to the matrix vector products at each round $1,\ldots,T$. We give
simple conditions on the perturbation vectors $g_1,\ldots,g_t$ such that under
these conditions, perturbed powering of a matrix $A \in \mathbb{R}^{n\times
n}$ for $O(\log n)$ rounds results in a vector $x$ such that: $\|Ax\|/\|x\|
\geq (1-\beta)\sigma_1(A)$ where $\sigma_1(A)$ is the top singular value of
$A$. Using this general analysis, we are then free to choose the perturbations
appropriately to guarantee differential privacy. The accuracy bounds we obtain
are a function of the scale of the noise that is necessary for privacy.

It is immediate that the magnitude of the perturbation that must be used to
guarantee differential privacy (of the matrix) when computing a matrix vector
product is proportional to the magnitude of the largest coordinate in the
vector. To prove our accuracy guarantees, therefore, it suffices to bound the
maximum magnitude of any coefficient of any of the vectors $x_1,\ldots,x_T$
that emerge during the steps of power iteration. Of course, if the matrix is
incoherent, then each $x_t$ can be written as a linear combination of basis
vectors that each have small coordinates $x_t = \sum_{i=1}^n \alpha_iv_i.$
Unfortunately this does not suffice to guarantee that $x_t$ will have small
coordinates without incurring a blow-up that depends on the number of nonzero
coefficients. However, we show that at each round,
$\mathrm{sign}(\alpha_1),\ldots,\mathrm{sign}(\alpha_n)$ are independent,
unbiased $\signs$ random variables. This, together with the incoherence assumption,
is enough to complete the analysis.

Finding a unit vector $x$ such that $\|Ax\| \geq (1-\beta)\sigma_1(A)$ is
sufficient to compute an accurate rank-$1$ approximation to $A$ in spectral
norm. If $x$ was exactly equal to the top singular vector of $A$, we could
then recurse, and compute the top singular vector of $A' = A - \sigma_1xx^T$,
from which we could compute an optimal rank $2$ approximation to $A$.
Unfortunately, $x$ is only an approximation to the top singular vector.
Therefore, in order to be able to usefully recurse on $A' = A -
\hat{\sigma_1}xx^T$, we require two conditions: (1) That $\|A'\|_2 \approx
\sigma_2(A)$, and (2) that $A'$ is nearly as incoherent as $A$. Condition (1) has
already been shown by Kapralov and Talwar \cite{KapralovT13}. Therefore, it
remains for us to show condition (2). We show that indeed the incoherence of the
matrix cannot increase by more than a factor of $\sqrt{r}$, where $r$ is the
rank of $A$, during any number of ``deflation'' steps. However, we do not know
whether this factor of $\sqrt{r}$ is necessary, or is merely an artifact of
our analysis. We leave removing this factor of $\sqrt{r}$ from our
approximation factor for computing rank-$k$ approximations when $k \geq 2$ as
an intriguing open problem.

Finally, we give a pointwise lower bound that shows that (up to log factors),
our algorithm for privately computing singular vectors is tight \emph{for
every setting of the coherence parameter}. We do this by reducing to
reconstruction lower bounds of Dinur and Nissim \cite{DinurN03}. Specifically,
we show, for every coherence parameter $C$, how to construct a matrix with
coherence $C$ from some private bit-valued database $D$ such that improving on
the performance of our algorithm would imply that an adversary would be able
to reconstruct $D$. Since reconstruction attacks are precluded by reasonable
values of $\epsilon$ and $\delta$, a lower bound for all $(\epsilon,\delta)$
private algorithms follows.

\subsection{Related Work}
There is by now an extensive literature on a wide variety of differentially private computations, which we do not attempt to survey here. Instead we focus on only the most relevant recent work.

There are several papers that consider the problem of privately approximating
the singular vectors of a matrix without any assumptions on the data. Blum et
al. \cite{BlumDMN05} first studied this problem, and gave a simple ``input
perturbation'' algorithm based on adding noise directly to the covariance
matrix. Chaudhuri et al \cite{ChaudhuriSS12} and Kapralov and Talwar \cite{KapralovT13} give matching worst-case
upper and lower bounds for privately computing the top eigenvector of a matrix
under the constraint of $(\epsilon,0)$-differential privacy: They achieve additive error $O(n/\epsilon)$. Both algorithms involve sampling a singular vector from the exponential mechanism. \cite{KapralovT13} also give a polynomial time algorithm for performing this sampling from the exponential mechanism, whereas \cite{ChaudhuriSS12} give a heuristic, but practical implementation using Markov-Chain Monte-Carlo. Our algorithm matches these worst case bounds, and also gives worst case bounds for $(\epsilon,\delta)$-privacy, with error $O(\sqrt{n}/\epsilon)$. In the event that the matrix has low coherence, we improve substantially over the worst case bounds. Moreover, we give the first analysis of a natural, efficient algorithm for this problem. Indeed, our algorithm is simply a variant on the classic power iteration method, and runs in time nearly linear in the input sparsity.

Low coherence conditions have been recently studied in a number of papers for
a number of matrix problems, and is a commonly satisfied condition on
matrices. Recently, Candes and Recht \cite{CandesR09} and Candes and Tao
\cite{CandesT10} considered the problem of \emph{matrix completion}. Accurate matrix completion
is impossible for arbitrary matrices, but \cite{CandesR09,CandesT10} show the
remarkable result that it is possible under low coherence assumptions. Candes
and Tao \cite{CandesT10} also show that almost every matrix satisfies a low
coherence condition, in the sense that randomly generated matrices will be low
coherence with extremely high probability.

Talwalkar and Rostamizadeh recently used low-coherence assumptions for the
problem of (non-private) low-rank matrix approximation \cite{TalwalkarR10}.  They showed that under low-coherence assumptions
similar to those of \cite{CandesR09,CandesT10}, the spectrum of a matrix is in fact
well approximated by a small number of randomly sampled columns, and give
formal guarantees on the approximation quality of the sampling based
Nystr\"{o}m method of low-rank matrix approximation.

Most related to this paper is Hardt and Roth \cite{HardtR12}, which gives an
algorithm for giving a rank-$k$ approximation to a private matrix $A$ in the
\emph{Frobenius} norm, where the approximation quality also depends on a
(slightly different) notion of matrix coherence. This work differs from
\cite{HardtR12} in several respects. First, a matrix may not have any good
approximation in the Frobenius norm (and hence the bounds of \cite{HardtR12}
might be vacuous), but still might have an excellent approximation in the spectral
norm. Second, \cite{HardtR12} does not give any means to actually compute the top
singular vector of the private matrix, and hence cannot be easily used for
applications (such as PCA, or spectral clustering) that require direct access
to the singular vector itself. Moreover, unlike in this paper, \cite{HardtR12}
do not show that their dependence on the coherence is tight---only that their
guarantees surpass any data-independent worst case guarantees. The bounds of
\cite{HardtR12} also incur a constant \emph{multiplicative} error, in addition
to an additive error. In this paper, we are able to avoid any multiplicative
error. Finally, the bounds of \cite{HardtR12} depend on the rank of the
private matrix $A$, a
dependence that we are able to remove when computing the top singular vector
of $A$, as well as a rank~$1$ approximation of $A$.

Related to the problem of approximating the spectrum of a matrix is the problem of approximating \emph{cuts} in a graph. This problem was first considered by Gupta, Roth, and Ullman \cite{GuptaRU12} who gave methods for efficiently releasing synthetic data for graph cuts with additive error $O(n^{1.5})$. Blocki et al \cite{BlockiBDS12} gave a method which achieves improved error for small cuts, but does not improve the worst-case error. Improving these bounds to the information theoretically optimal bound of $O(n\log n)$ via an efficient algorithm remains an interesting open question. Note that smaller error is efficiently achievable for a polynomial number of cut queries, using private multiplicative weights \cite{HardtR10} or randomized response.
\subsection*{Acknowledgments}
We would like to thank Frank McSherry for suggesting the use of Power
Iteration.
\section{Preliminaries}
\sectionlabel{prelim}
We view our dataset as a real valued \emph{matrix} $A\in\mathbb{R}^{m\times
n}.$
%
\begin{definition} We say that two matrices $A,
A' \in \mathbb{R}^{m\times n}$ are \emph{neighboring} if $A - A' = \alpha e_se_t^T$
where $e_s,e_t$ are two standard basis vectors and $\alpha\in[-1,1].$ In other
words $A$ and $A'$ differ in precisely one entry by at most $1$ in absolute
value.
\end{definition}
We use the by now standard privacy solution concept of differential privacy:
\begin{definition}
An algorithm $M\colon\mathbb{R}^{m\times n}\rightarrow R$ (where $R$ is some
arbitrary abstract range) is \emph{$(\epsilon,\delta)$-differentially private}
if for all pairs of neighboring databases $A, A' \in \mathbb{R}^{m\times n}$, and for
all subsets of the range $S \subseteq R$ we have
$\Pr\Set{M(A) \in S} \leq
\exp(\epsilon)\Pr\Set{M(A') \in S} + \delta\mper$
\end{definition}

We make use of the following useful facts about differential privacy.
\begin{fact}
If $M:\mathbb{R}^{m\times n}\rightarrow R$ is $(\epsilon,\delta)$-differentially private, and $M':R\rightarrow R'$ is an arbitrary randomized algorithm mapping $R$ to $R'$, then $M'(M(\cdot)):\mathbb{R}^{m\times n}\rightarrow R'$ is $(\epsilon,\delta)$-differentially private.
\end{fact}

The following useful theorem of Dwork, Rothblum, and Vadhan tells us how differential privacy guarantees compose.
\begin{theorem}[Composition \cite{DworkRV10}]
\label{thm:composition}
Let $\epsilon,\delta\in(0,1),\delta'>0.$
If $M_1, \ldots, M_k$ are each $(\epsilon,\delta)$-differentially private
algorithms, then the algorithm $M(A)
\equiv (M_1(A),\ldots,M_k(A))$ releasing the concatenation of the results of
each algorithm is $(k\epsilon, k\delta)$-differentially private. It is also $(\epsilon', k\delta + \delta')$-differentially private for
$\epsilon' < \sqrt{2k\ln(1/\delta')}\epsilon + 2k\epsilon^2.$
\end{theorem}

We denote the $1$-dimensional Gaussian distribution of mean $\mu$ and
variance $\sigma^2$ by $N(\mu,\sigma^2).$ We use $N(\mu,\sigma^2)^d$ to denote
the distribution over $d$-dimensional vectors with i.i.d. coordinates sampled from
$N(\mu,\sigma^2).$ We write $X\sim D$ to indicate that a
variable $X$ is distributed according to a distribution~$D.$ We note
the following useful fact about the Gaussian distribution.
\begin{fact}\factlabel{gaussian-sum}
If $g_i\sim N(\mu_i,\sigma_i^2),$ then
$\sum g_i \sim N\left(\sum_i\mu_i,\sum_i\sigma_i^2\right)\mper$
\end{fact}
The following theorem is well known folklore.
\begin{theorem}[Gaussian Mechanism]
\label{thm:gaussian}
Let $\epsilon>0,$ $\delta\in(0,1/2).$
Let $u, v \in \mathbb{R}^d$ be any two vectors such that $\|u-v\|_2 \leq c$.
Put $\sigma = 4c\epsilon^{-1}\sqrt{\log(2/\delta)}$. Then, for every
measurable set $A\subseteq\R^d$ and $g\sim N(0,\sigma^2)^d,$ we have
$\exp(-\epsilon)\Pr\{v + g \in A\} - \delta
\le
\Pr\Set{u + g \in A} \leq \exp(\epsilon)\Pr\{v + g \in A\} + \delta\mper$
\end{theorem}

\paragraph{Vector and matrix norms.} We denote by $\|\cdot\|_p$ the
$\ell_p$-norm of a vector and sometimes use $\|\cdot\|$ as a shorthand for the
Euclidean norm. Given a real $m\times n$ matrix $A,$ we will work with the
\emph{spectral norm} $\|A\|_2$ and the Frobenius norm~$\|A\|_F$ defined as
\begin{equation}
\|A\|_2  \defeq \max_{\|x\|=1}\|Ax\|\qquad\text{and}
\qquad \|A\|_F \defeq \sqrt{\sum_{i,j}a_{ij}^2}\mper
\end{equation}
For any $m\times n$ matrix $A$ of rank $r$ we have
$\|A\|_2\le\|A\|_F\le\sqrt{r}\cdot\|A\|_2\mper$
%
%
%
%
\paragraph{Singular Value Decomposition.} Given a matrix $A\in\R^{m\times n}.$
The \emph{right singular vectors} of $A$ are the eigenvectors of $A^TA.$ The
\emph{left singular vectors} of $A$ are the eigenvectors of $AA^T.$ The singular
values of $A$ are denoted by $\sigma_i(A)$ and defined as the square root of
the $i$-th eigenvalue of $A^TA.$ The singular value decomposition is any
decomposition of $A$ satisfying $A=U\Sigma V^T$ where $U\in\R^{m\times
m},V\in\R^{n\times n}$ are unitary matrices and $\Sigma\in\R^{m\times n}$
satisfies $\Sigma_{ii}=\sigma_i(A)$ and $\Sigma_{ij}=0$ for $i\ne j.$ The
colums of $U$ are the left singular vectors of $A$ and the columns of $V$ are
the right singular vectors of $A.$

\subsection{Matrix coherence}
We will work with the following standard notion of \emph{coherence} throughout the
paper.
\begin{definition}[$\mu$-Coherence]
Let $A\in\mathbb{R}^{m\times n}$ with $m\le n$ be a symmetric real matrix with
a given singular value decomposition $A=U\Sigma V^t.$
We define the \emph{$\mu$-coherence} of $A$ with respect to $U$ and $V$ as
\[
\mu(A) \defeq \max\Set{m\|U\|_\infty^2,n\|V\|_\infty^2}\mper
\]
Note that $1\le\mu(A)\le n.$
\end{definition}
We remark that the coherence of $A$ is defined with respect to a particular singular value
decomposition since the SVD is in general not unique.

\subsection{Reduction to symmetric matrices}
\sectionlabel{symmetric}

Throughout our work we will restrict our attention real symmetric $n\times
n$ matrices. All of our results apply, however, more generally to asymmetric
matrices. Indeed, given $A\in\R^{m\times n}$ with SVD $A=\sum_{i=1}^r \sigma_i
u_iv_i^T$ and rank $r,$ we can instead consider the symmetric $(m+n)\times(m+n)$ matrix
\[
B=\left[
\begin{array}{cc}
0 & A \\
A^T & 0
\end{array}
\right]\mper
\]
The next fact summarizes all properties of $B$ that we will need.
\begin{fact}
The matrix $B$ has the following properties:
$B$ has a rank $2r$ and singular values $\sigma_1,\dots,\sigma_r$ each
occuring with multiplicity two.
The singular vectors corresponding to a singular value $\sigma$ are spanned
by the vectors $\Set{(u_i,0), (0,v_i)\colon \sigma_i=\sigma}.$
An entry change in $A$ corresponds to two entry changes in $B.$
Furthermore, $\mu(B)=\mu(A).$
\end{fact}
In particular, this fact implies that an algorithm to find the singular vectors of $B$ will
also recover the singular vectors of $A$ up to small loss in the parameters.
Moreover, an algorithm that achieves $(\epsilon/2,\delta/2)$-differential privacy on $B$
is also $(\epsilon,\delta)$-differentially private with respect to $A.$


%

\section{Robust convergence of power iteration}

In this section we analyze a generic variant of power iteration in which a
perturbation is added to the computation at each step. The noise vector can be
chosen adaptively and adversarially in each round.
We will derive general conditions under which power iteration converges.

\begin{figure}[h]
\begin{boxedminipage}{\textwidth}
\noindent \textbf{Input:} Matrix $A\in\mathbb{R}^{n\times n},$ number of
iterations $T\in\mathbb{N},$ parameter $\beta\in(0,1),$
\begin{enumerate}
\item Let $x_0$ be unit vector.
\item For $t = 1$ to $T$:
\begin{enumerate}
  \item Let $g_t$ be an arbitrary perturbation.
  \item Let $x'_t = Ax_{t-1} + g_t$,
  \item If $\|x'_t\|\ge(1-\beta)\sigma_1,$ then terminate and output $x_{t-1}.$
  \item Otherwise let $x_t = \frac{x'_t}{\|x'_t\|_2},$ and continue.
\end{enumerate}
\end{enumerate}
\noindent \textbf{Output:} Vector $x_T\in\mathbb{R}^n$ unless the algorithm
terminated previously.
\end{boxedminipage}
\caption{Power iteration with adversarial noise}
\figurelabel{robust-powering}
\end{figure}

\begin{lemma}[Robust Convergence]
\label{lem:robust}
Let $A$ be a matrix such that $\sigma_{k+1}(A)\le(1-\gamma)\sigma_k(A)$ for
some $k<n$ and $\gamma>0.$ Let $U$ be the space spanned by the top $k$
singular vector of $A,$ let $V$ be the space spanned by the last $n-k$ singular
vectors.  Further assume that there are numbers $\Delta,\Delta_U,\Delta_V>0$
such that the following conditions are met:
\begin{enumerate}
\item \itemlabel{upper}
For all $t,$ $\|g_t\|\le\Delta,$
$\|P_Ug_t\|\le\Delta_U$ and $\|P_Vg_t\|\le\Delta_V.$
\item \itemlabel{lower}
$\|P_Ux_0\| \ge \frac{8\Delta_U}{\gamma\sigma_k(A)}$
and
$\|P_Vx_0\| \ge \frac{8\Delta_V}{\gamma\sigma_k(A)}$
\item \itemlabel{sigma}
$\sigma_k(A)\ge 9\Delta/\beta\gamma,$ for some $0<\beta<1.$
\end{enumerate}
Then, for $T=4\log(\sigma_k(A)),$
the algorithm outputs a vector $x\in\R^n$
such that
\[
\frac{\|Ax\|}{\|x\|}\ge(1-\beta)\sigma_k(A)\mper
\]
\end{lemma}
\begin{proof}
Put $\sigma = \sigma_k(A)$ and note that by assumption
$\sigma_{k+1}=(1-\gamma)\sigma_k$ for some $\gamma>0.$
We will consider the potential function
\[
\Psi_t = \frac{\|P_V x_t\|}{\|P_U x_t\|}.
\]
Suppose that in some round $t,$ we have
\begin{equation}\equationlabel{cond}
\sigma\|P_Vx_{t-1}\| \ge \frac{8\Delta_V}\gamma\quad\text{and}\quad
\sigma\|P_Ux_{t-1}\| \ge \frac{8\Delta_U}\gamma\mper
\end{equation}
We note that by our assumption on the matrix, these conditions are met in the
first round $t=1$ as a consequence of \itemref{lower}.
Let us derive an expression for the potential drop in round $t$ under the
above assumption. We have, using \itemref{upper},
\[
\frac{\|P_V x_t\|}{\|P_U x_t\|}
= \frac{\|P_V (Ax_{t-1} + g_t)\|}{\|P_U (A x_{t-1} + g_t)\|}
\le
\frac{\|P_V Ax_{t-1}\| + \|P_Vg_t\|}{\|P_UAx_{t-1}\| - \| P_Ug_t\|}
\le \frac{(1-\gamma)\sigma\|P_V x_{t-1}\| + \Delta_V}{\sigma\|P_Ux_{t-1}\| -
\Delta_U}
\]
By the assumption in \equationref{cond}, we have
\[
\frac{(1-\gamma)\sigma\|P_V x_{t-1}\| +\Delta_V}{\sigma\|P_Ux_{t-1}\|
-\Delta_U}
\le
\frac{(1-7\gamma/8)\sigma\|P_V x_{t-1}\|}{(1-\gamma/8)\sigma\|P_Ux_{t-1}\|}
\le \left(1-\frac\gamma2\right)
\frac{\|P_V x_{t-1}\|}{\|P_Ux_{t-1}\|}
= \left(1-\frac\gamma2\right)\Psi_{t-1}
\]
We furthermore claim that if the conditions in \equationref{cond} hold true in
round $t,$ then we must have
$\|P_Ux_t\|\ge\|P_Ux_{t-1}\|\mper$
This follows from our previous analysis, because $\Psi_t\le\Psi_{t-1}$
but
\[
1=\|x_t\|=\sqrt{\|P_Ux_t\|^2+\|P_Vx_t\|^2}\mper
\]
This in particular means that if
the conditions are true in round $t,$ then the second condition in
\equationref{cond} continues to be true in round $t+1,$ and only the first
condition can fail. At this point we distinguish two cases.

\paragraph{Case 1.}
Suppose there is a round where the $t\le T,$ where the first condition
fails to hold. Let $t^*$ be the smallest such round and put $x=x_{t^*-1}.$
By the previous argument, in this round we must have
\begin{equation}\equationlabel{fail}
1=\|x\|^2=\|P_Vx\|^2+\|P_Ux\|^2
 \le \left(\frac{8\Delta_V}{\gamma\sigma}\right)^2+\|P_Ux\|^2\mper
\end{equation}
From this we conclude that
$\|P_Ux\|\ge\sqrt{1-(8\Delta_V/\gamma\sigma)^2}\ge1-8\Delta_V/\gamma\sigma.$
Hence,
\[
\|Ax_{t^*-1}+g_{t^*}\|\ge\|AP_Ux\|-\|g_t\|
\ge \left(1-\frac{8\Delta_V}{\gamma\sigma}\right)\sigma - \Delta
\ge \left(1-\frac{8\Delta}{\gamma\sigma}-\frac{\Delta}\sigma\right)\sigma
\ge \left(1-\frac{9\Delta}{\gamma\sigma}\right)\sigma
\]
Here we used that $\Delta_V\le\Delta$ which is without loss of generality.
Therefore, using \itemref{sigma},
\[
\|x_{t^*}'\|
\ge \left(1-\frac{9\Delta}{\gamma\sigma}\right)\sigma
\ge (1-\beta)\sigma_1\mcom
\]
This means that the algorithm terminates in round $t^*$ and outputs
$x_{t^*-1},$ which satisfies the conclusion of the lemma.

\paragraph{Case 2.}
Suppose there is no round $t\le T,$ where
\equationref{cond} fails. By our potential argument and the choice of $T,$
this means that
\[
\Psi_T\le\left(1-\frac\gamma2\right)^T\Psi_0
\le \frac{\exp(-\gamma T/2)}{\|P_Ux_0\|}
\le \frac{\gamma\sigma}{8\Delta_U}\exp(-\gamma T/2)
= \frac{\gamma}{8\Delta_U\sigma}
\le \beta
\]
In particular, $x+T$ satisfies $\|P_Vx_T\|\le \beta\|P_Ux_T\|\le
\beta.$ Thus, $\|P_Ux_T\|\ge \sqrt{1-\beta^2}$ and
$\|Ax_T\|\ge(1-\beta)\sigma.$ This show that $x_T$ satisfies
the conclusion of the lemma.
\end{proof}

The next corollary states a variant of \lemmaref{robust} where we express all
conditions in terms of $\sigma_1(A)$ rather than $\sigma_k(A).$

\begin{corollary}
\corollarylabel{robust}
Let $\alpha\in(0,1).$
Let $A$ be a matrix such that $\sigma_{k+1}(A)\le(1-\gamma/2)\sigma_1(A)$ for
some $k<n.$ Let $U$ be the space spanned by the top $k$
singular vector of $A,$ let $V$ be the space spanned by last $n-k$ singular
vectors. Further assume that there are numbers $\Delta,\Delta_U,\Delta_V>0$
such that the following conditions are met:
\begin{enumerate}
\item
For all $t,$ $\|g_t\|\le\Delta,$
$\|P_Ug_t\|\le\Delta_U$ and $\|P_Vg_t\|\le\Delta_V.$
\item
$\|P_Ux_0\| \ge \frac{32k\Delta_U}{\gamma(1-\gamma)\sigma_1(A)}$
and
$\|P_Vx_0\| \ge \frac{32k\Delta_V}{\gamma(1-\gamma)\sigma_1(A)}$
\item
$\sigma_1(A)\ge \frac{72k\Delta}{\beta\gamma(1-\gamma)}.$
\end{enumerate}
Then, for $T=4\log(\sigma_1(A)),$
the algorithm outputs a vector $x\in\R^n$
such that
\[
\frac{\|Ax\|}{\|x\|}\ge(1-\beta)\sigma_1(A)\mper
\]
\end{corollary}
\begin{proof}
We claim that there exists a $k'\le k$ such that $\sigma_{k'}(A)\le
(1-\gamma/4k)\sigma_1(A).$
Indeed, if this is not the case then
\[
\sigma_k(A)\ge \prod_{i=1}^k\left(1-\frac{\gamma}{4k}\right)\sigma_1(A) >
\left(1-\frac{\gamma}{2}\right)\sigma_1(A),
\]
thus violating the assumption of the lemma.
Moreover, $k'$ satisfies $\sigma_{k'}(A)\ge(1-\frac{\gamma}{2})\sigma_1(A).$
We will thus apply \lemmaref{robust} to this $k'$ setting
$\gamma'=\gamma/4k.$ It is easy to verify that by our assumptions above, the
conditions of \lemmaref{robust} are satisfied. Hence, the output $x$ of the
algorithm satisfies
\[
\frac{\|Ax\|}{\|x\|} \ge (1-\frac{\gamma}{2})\sigma_{k'}(A)\ge
(1-\gamma/2)^2\sigma_1(A)\ge(1-\gamma)\sigma_1(A).
\]
\end{proof}

\begin{remark}
We will typically need $k$ in \corollaryref{robust} to be relatively small compared to $n.$
We think of this as a mild assumption even when $k$ and $\alpha$ are constant.
In particular, it is implied by the assumption that
$A$ has a good low-rank approximation for small $k.$ Indeed, if $\sigma_{k+1} >
(1-\alpha)\sigma_1,$ then the best rank $k$ approximation to $A$ has spectral
error $(1-\alpha)\|A\|_2.$
\end{remark}

\subsection{Privacy-Preserving Power Iteration}
\sectionlabel{PPI}

We will next turn the robust power iteration algorithm from the
previous section into a privacy-preserving version. The algorithm is outlined
below.

\begin{figure}[h]
\begin{boxedminipage}{\textwidth}
\noindent \textbf{Input:} Matrix $A\in\mathbb{R}^{n\times n},$ number of
iterations $T\in\mathbb{N},$ privacy parameters $\epsilon,\delta>0,$ upper bound on
coherence $C>0.$
\begin{enumerate}
\item
Let $\sigma = 2\epsilon^{-1}{\sqrt{4T\log(1/\delta)}}.$
\item  Let $x_0=g_0\sim N(0, 1/n)^n.$
\item For $t = 1$ to $T$:
\begin{enumerate}
  \item If $\|x_{t-1}\|_\infty^2>C/n,$ terminate and output ``fail''.
  \item Let $g_t \sim  N\left(0, \frac {C\sigma^2}{n}\right)^n$
  \item Let $x'_t = Ax_{t-1} + g_t$
  \item Put $x_t = \frac{x'_t}{\|x'_t\|_2}$
\end{enumerate}
\end{enumerate}
\noindent \textbf{Output:} Vector $x_T\in\mathbb{R}^n$
\end{boxedminipage}
\caption{Private power iteration (PPI)}
\figurelabel{powering}
\end{figure}

\begin{lemma}
The algorithm PPI satisfies $(\epsilon,\delta)$-differential privacy.
\end{lemma}
\begin{proof}
By \theoremref{gaussian}, the algorithm satisfies
$(\epsilon',\delta)$-differential privacy in each round.
Here, $\epsilon'$ was chosen small enough so that \theoremref{composition}
implies $(\epsilon,\delta)$-differential privacy for the algorithm over all.
\end{proof}

The next lemma states the guarantees of the algorithm assuming that it
successfully terminates.

\begin{lemma}\lemmalabel{utility}
Let $\alpha>0.$
Let $A$ be a matrix satisfying $\sigma_k\le (1-\gamma/2)\sigma_1$ for some
$k\ge1.$
Put $T=4\log(\sigma_1(A)).$
Further assume that for some $\beta \geq 0$, $A$ satisfies
\begin{equation}
\equationlabel{lower}
\|A\|_2 =  \frac{\Theta Tk \sqrt{C\log(n)\log(1/\delta)}}{\epsilon\gamma\beta}\mper
\end{equation}
for some sufficiently large constant $\Theta>0.$
Assume that PPI terminates successfully and outputs $x_T$ on input of $A,$
$T,$ and $C.$ Then, with probability $9/10,$
\[
\|Ax_T\|\ge (1-\beta)\|A\|_2\mper
\]
\end{lemma}
\begin{proof}
Our goal is to apply \corollaryref{robust}. For this we need to verify that $A$
and $g_t$ satisfy various assumptions of the lemma.
Put $\Delta=\sqrt{4C\log(n)}\sigma.$
With this choice of $\Delta,$ we have by basic Gaussian concentration bounds
(see \lemmaref{gaussian-norm}):
\begin{enumerate}
\item $\Pr\Set{\|g_t\|>\Delta}\le 1/n^2.$
\item $\Pr\Set{\|P_Ug_t\|>\sqrt{\frac kn}\Delta}\le 1/n^2.$
\item $\Pr\Set{\|P_Vg_t\|>\sqrt{\frac {n-k}n}\Delta}\le 1/n^2.$
\end{enumerate}
Hence, with probability $1-1/n,$ none of these events occur for any $t\in[T].$
This verifies that the first assumption of \corollaryref{robust}
holds with high probability for this setting of $\Delta.$
Further note that, by Gaussian anti-concentration bounds (as stated in
\lemmaref{gaussian-norm}) the following claims are true:
\begin{enumerate}
\item $\Pr\Set{\|P_Ux_0\| \ge \sqrt{\frac{k}{50en}}}\ge 98/100$
\item $\Pr\Set{\|P_Vx_0\| \ge \sqrt{\frac{n-k}{50en}}}\ge 98/100$
\end{enumerate}
Hence, both of these events occur with probability $96/100.$
On the other hand the second condition of \corollaryref{robust} requires that
$\|P_Ux_0\|\ge O(k\Delta_U/\gamma\sigma_1(A)).$
Assuming the event $\|P_Ux_0\|\ge\sqrt{k/100n}$ occurred
this corresponds to a lower bound of the form $\sigma_1(A)\ge
O(k\Delta/\gamma)$ which is satisfied by \equationref{lower}. The analogous
argument applies to $\|P_Vx_0\|.$
Finally, the third condition of \lemmaref{robust} follows by comparison with
\equationref{lower}. Hence, the lemma follows.
\end{proof}

\section{Power Iteration and Incoherence}
\label{sec:incoherence}

We will next establish an important symmetry property of the algorithm.
Specifically, we will show that for any of the eigenvectors $u$ of $A$ (assuming
$A$ is symmetric), the sign of the correlation between $u$ and ny intermediate
vector $x_t,$ i.e. $\sign(\langle u,x_t\rangle)$ is unbiased and independent
$\sign(\langle v,x_t\rangle)$ for any other eigenvector $v.$ This property is
rather obvious in the noise-free case where $x_t$ is simply proportional to
$A^tx_0.$ Hence, the sign of $\langle u,x_t\rangle$ is determined by the sign
of $\langle u,x_0\rangle.$ Intuitively, the property continues to hold in the noisy case, because
the noise that we add is symmetric.

\begin{lemma}[Sign Symmetry]
\lemmalabel{signs}
Let $A$ be a symmetric matrix given in its eigendecomposition as
$A=\sum_{i=1}^n\sigma_i u_i u_i^T.$
Let $t\ge 0$ and put $X_i = \sign(\langle u_i,x_t\rangle)$ for $i\in[n].$
Then $(X_1,\dots,X_n)$ is uniformly distributed in $\signs^n.$
\end{lemma}

\begin{proof}
We will establish by induction on~$t$ that the following two conditions
hold for every $t\ge0:$
\begin{enumerate}
\item $Y_i(t) = \langle u_i,x_t\rangle$ is a symmetric random variable
\item $\sign(Y_i(t))$ is independent of $Y_j(t)$ for all $j\ne i.$
\end{enumerate}
Observe that these two conditions imply the statement of the lemma.
In the base case notice that $Y_i(0)$ is just a random Gaussian variable
$N(0,1/n)$ and hence symmetric. Now, let $t\ge 1$ and consider
\[
Y_i(t) = \frac{\langle u_i, Ax_{t-1} + g_t\rangle }{\|Ax_{t-1}+g_t\|}
= \frac{\sigma_i \langle u_i,x_{t-1}\rangle + \langle u_i, g_t\rangle }{\|Ax_{t-1}+g_t\|}
= \frac{\sigma_i Y_i(t-1) + \langle u_i, g_t\rangle }{\|Ax_{t-1}+g_t\|}\mper
\]
Let $D_i=\sigma_i Y_i(t-1)+\langle u_i,g_t\rangle.$ Notice that $D_i$ is a
symmetric random variable, since it is the sum of two independent symmetric
random variable. Here we used the induction hypothesis on $Y_i(t-1).$ We can
see that $Y_i(t)$ is a rescaling of a symmetric random variable, but we also
need to show that the rescaling is independent of $\sign(D_i).$
Note that
\[
\|Ax_{t-1}+g_t\|
= \left\|\sum_{j=1}^n u_j\left(\sigma_i\langle u_j,x_{t-1}\rangle + \langle
u_j,g_t\rangle\right)\right\|
= \sqrt{\sum_{i=1}^n D_i^2}
\]
This shows that the normalization term can be computed from $D_i^2$ and
$\sigma_j Y_j(t-1)+\langle u_j,g_t\rangle$ for $j\ne i.$ Note that each of
these terms is independent of $\sign(D_i).$ Here we used the induction
hypothesis on $Y_j(t-1)$ and the fact that $\langle u_j,g_t\rangle$ are
independent Gaussians for all $j\in[n].$ We conclude that $Y_i(t)=
\frac{D_i}{\sqrt{D_1^2+\cdots+D_n^2}}$ is a symmetric random variable.

It remains to show that $\sign(Y_i(t))$ is independent of $Y_j(t),$ for all
$j\ne[i].$ We have already shown that the normalization term appearing in
$Y_j(t)$ is statistically independent of $\sign(Y_i(t)).$ Moreover, by
induction hypothesis, the numerator $\sigma_j Y_j(t-1)+\langle u_j,g_t\rangle$
is statistically independent of $\sign(Y_i(t-1))$ and statistically
independent of $\langle u_i,g_t\rangle.$ In particular, conditioning on any
subset of the variables $Y_j(t),j\ne i$ leaves the two variables
$\sign(Y_i(t-1))$ and $\sign(\langle u_i,g_t\rangle)$ unbiased. This implies
that no matter what the value of $|Y_i(t-1)|$ and $|\langle u_i,g_t\rangle|$
is, the variable $\sign(Y_i(t))$ is unbiased.
\end{proof}

We will use the previous lemma to bound the $\ell_\infty$-norm of the
intermediate vectors $x_t$ arising in power iteration in terms of the
coherence of the input matrix. We need the following large deviation bound.

\begin{lemma}
\lemmalabel{sign-deviation}
Let $\alpha_1,\dots,\alpha_n$ be scalars such that $\sum_{i=1}^n\alpha_i^2=1$
and $u_1,\dots,u_n$ are unit vectors in $\R^n.$ Put
$B=\max_{i=1}^n\|u_i\|_\infty.$
Further let $(s_1,\dots,s_n)$ be chosen
uniformly at random in $\signs^n.$
Then,
\[
\Pr\Set{ \left\|\sum_{i=1}^ns_i\alpha_iu_i\right\|_\infty >
4B\sqrt{\log n}} \le 1/n^3\mper
\]
\end{lemma}

\begin{proof}
Let $X = \sum_{i=1}^n X_i$ where $X_i = s_i\alpha_iu_i.$ We will bound the
deviation of $X$ in each entry and then take a union bound over all entries.
Consider $Z=\sum_{i=1}^n Z_i$ where $Z_i$ is the first entry of $X_i$. The
argument is identical for all other entries of $X.$
We have $\E Z = 0$ and $\E Z^2 = \sum_{i=1}^n \E Z_i^2 \le
B^2\sum_{i=1}^n\alpha_i^2=B^2.$  Hence, by \theoremref{chernoff} (Chernoff
bound),
\[
\Pr\Set{ \left|Z\right| > 4B\sqrt{\log(n)}} \le
\exp\left(-\frac{16B^2\log(n)}{4B^2}\right)
\le\exp(-4\log(n))=\frac1{n^4} \mper
\]
The claim follows by taking a union bound over all $n$ entries of $X.$
\end{proof}

\begin{lemma}
\lemmalabel{termination}
Let $A\in\R^{n\times n}.$ Suppose PPI is invoked on $A,$ $T\le n,$ and $C\ge
16\mu(A)\log(n)$ and any choice of $\epsilon,\delta>0.$ Then, with
probability $1-1/n,$ the algorithm terminates successfully after round~$T.$
\end{lemma}
\begin{proof}
The only way for the algorithm to terminate prematurely in step $t+1$
is if the vector $x_t$ satisfies $\|x_t\|_\infty\ge 4\sqrt{\mu(A)\log(n)/n}.$
We will argue that this happens only with probability $1/n^2.$ Hence, by
taking a union bound over all rounds $T\le n,$ we conclude that the algorithm
must terminate with probability $1-1/n.$

Indeed, let $A=\sum_{i=1}^n\sigma_i u_i u_i^T$ be given in its
eigendecomposition. Note that
$B=\max_{i=1}^n\|u_i\|_\infty\le\sqrt{\mu(A)/n}.$
On the other hand, we can write $x_t=\sum_{i=1}^ns_i\alpha_iu_i$ where
$\alpha_i$ are non-negative scalars such that $\sum_{i=1}^n\alpha_i^2=1,$ and
$s_i\in\signs.$ Notice that $s_i=\sign(\langle x_t,u_i\rangle).$ Hence, by
\lemmaref{signs}, the signs $(s_1,\dots,s_n)$ are distributed uniformly at
random in $\signs^n.$ Hence, by \lemmaref{sign-deviation}, it follows that
\[
\Pr\Set{ \left\|x_t\right\|_\infty > 4B\sqrt{\log n}} \le 1/n^3\mper
\]
Hene, a union bound over all $t\in[T]$ completes the proof.
\end{proof}

Finally, we can combine \lemmaref{utility} and \lemmaref{termination}
to conclude that private power iteration converges does not terminate
prematurely and the output vector gives the desired error bound.
\begin{theorem}\theoremlabel{powering}
Let $\gamma,\beta>0.$ Let $A$ be a matrix satisfying $\sigma_k\le (1-\gamma/2)\sigma_1$ for some
$k\ge1.$ Put $T=4\log(\sigma_1(A)).$
Further assume $A$ satisfies
\begin{equation}
\equationlabel{lower}
\|A\|_2= \frac{\Theta Tk
\sqrt{\mu(A)\log(1/\delta)}\log(n)}{\epsilon\gamma\beta}\mper
\end{equation}
for some sufficiently large constant $\Theta>0.$
Then, with probability $8/10,$ on input of $A,$ $T,$ $(\epsilon,\delta)$ and
$C\ge16\mu(A)\log(n),$ the algorithm PPI outputs a vector $x,$ such that
\[
\frac{\|Ax\|}{\|x\|}\ge (1-\beta)\|A\|_2\mper
\]
Equivalently:
\[
\|Ax_T\|\ge \sigma_1(A) - \frac{\Theta Tk \sqrt{C\log(n)\log(1/\delta)}}{\epsilon\gamma}\mper
\]
\end{theorem}
\begin{proof}
The proof follows directly by combining \lemmaref{utility} applied with
$C=16\mu(A)\log(n)$ with \lemmaref{termination}. The latter lemma implies
that for this setting the algorithm terminates with probability $1-1/n.$ The
former lemma implies that the stated error bound holds in this case with
probability $9/10.$ Both event occur simultaneously with probability
$9/10-o(1).$
\end{proof}

\begin{remark}[On choosing $T$ and $C$]
As stated \theoremref{powering} requires the input to the algorithm to depend
on two sensitive quantities, i.e., $\sigma_1(A)$ and $\mu(A).$ It is easy to
get rid of this using standard techniques. We can upper bound $\sigma_1(A)$ by
$\|A\|_1=\sum_{ij}|A_{ij}|$ which can be computed efficiently and privately
(as it is $1$-sensitive). Since the dependence on $\sigma_1(A)$ in the choice
of $T$ is only logarithmic, this can change the error bounds only by constant
factors. To get rid of $\mu(A),$ we can try all choices of $C=2^{i},$
$i\in\{0,1,\dots,\log(n)\}.$ Since $\mu(A)\le n,$ this process will eventually
find a setting of $C$ that gives the right upper bound up to an overestimate
of at most a factor $2.$ As we need to scale down $(\epsilon,\delta)$ by a
$\log(n)$ factor in each execution, the error bounds deteriorate by a
$O(\log(n))$-factor. This loss can be replaced by $O(\log\log n)$ using the
exponential mechanism~\cite{McSherryT07}. We omit the details as they are
standard.
\end{remark}

\section{Rank $k$ approximations and Deflation}
\label{sec:deflation}
In this section, we show how to successively call our algorithm for obtaining
rank $1$ approximations to obtain a rank $k$ approximation. To do this, we
need to argue two things. First, we must argue that approximately optimal rank
$1$ approximations to successively `deflated' versions of our original matrix
can be combined to yield an approximately optimal rank $k$ approximation.
Second, we must argue that incoherence is propagated throughout the deflation
process, so that we can in fact obtain good rank $1$ approximations to the
deflated matrices.

\begin{figure}[h]
\begin{boxedminipage}{\textwidth}
\noindent \textbf{Input:} Matrix $A\in\mathbb{R}^{n\times n},$ target rank $k$, number of
iterations $T\in\mathbb{N},$ privacy parameters $\epsilon,\delta>0,$ upper bound on
coherence $C>0.$
\begin{enumerate}
\item Let $\epsilon' = \epsilon/(\sqrt{4k \ln(1/\delta)})$, $\delta' = \delta/k$
\item Let $A_0 \leftarrow A$, $B_0 \leftarrow 0$.
\item For $i = 1$ to $k$:
\begin{enumerate}
  \item Let $v_{i} \leftarrow \textrm{PPI}(A_{i-1}, T, \epsilon', \delta', C)$
  \item Let $\hat{\sigma_{i}} = \|A_{i-1}v_i\|_2 + \textrm{Lap}(1/\epsilon')$
  \item Let $A_i \leftarrow A_{i-1} - \hat{\sigma_{i}}v_iv_i^T$, $B_i \leftarrow B_{i-1} + \hat{\sigma_{i}}v_iv_i^T$.
\end{enumerate}
\end{enumerate}
\noindent \textbf{Output:} Matrix $A_k$
\end{boxedminipage}
\caption{Rank $k$ approximation (rank-k).}
\figurelabel{rankk}
\end{figure}

Our analysis will be based on a useful lemma of Kapralov and Talwar, that
shows that the standard ``matrix deflation'' method can be applied even given
only approximate eigenvectors. The lemma here is actually an easy modification
of the lemma from \cite{KapralovT13}. The details can be found in Appendix \ref{sec:deflationappendix}

\begin{lemma}[Deflation Lemma \cite{KapralovT13}]
\label{cor:modified}
Let $A$ be a symmetric matrix with eigenvalues $\lambda_1 \geq \ldots \geq \lambda_n$. There exists a constant $C > 0$ so that the following holds. Let $x$ be any unit vector such that $\|Ax\| \geq (1-\alpha/C)\lambda_1$, where $\alpha \in (0,1)$. Let $A' = A - t v\cdot v^T$, where $t \in (1\pm \alpha/C)\|A x\|$. Denote the eigenvalues of $A'$ by $\lambda_1' \geq \ldots \geq \lambda_n'$.
\begin{enumerate}
\item $\lambda_k \leq \lambda'_{k-1} \leq \min(\lambda_{k-1}, \lambda_k+\alpha\lambda_1)$ for each $k \in \{1,\ldots,n\}$.
\end{enumerate}
\end{lemma}

We now argue that deflation preserves incoherence. Here, we make use of two
lemmas from \cite{HardtR12}.

\begin{definition}[$\mu_0$-coherence]
\definitionlabel{mu0}
Let $U$ be an $n\times r$ matrix with orthonormal columns and $r\le n.$
Recall, that $P_U=UU^T.$ The \emph{$\mu_0$-coherence} of $U$ is defined as
\begin{equation}
\mu_0(U)
= \frac nr\max_{1\le j\le n} \|P_Ue_j\|^2
= \frac nr\max_{1\le j\le n} \|U_{(j)}\|^2\mper
\end{equation}
Here, $e_j$ denotes the $j$-th $n$-dimensional standard basis vector and
$U_{(j)}$ denotes the $j$-th row of $U.$
The \emph{$\mu_0$-coherence} of an $n\times n$ matrix $A$ of rank $r$ given in
its singular value decomposition $U\Sigma V^T$ where $U\in\R^{n\times r}$ is
defined as $\mu_0(U).$
\end{definition}
Observe that we always have $\mu_0(A) \leq \mu(A).$

\begin{lemma}[\cite{HardtR12}]
\lemmalabel{mu0-perturb}
Let $u_1,\dots,u_r\in\R^n$ be orthonormal vectors. Pick unit vectors
$n_1,\dots,n_k\in\mathbb{S}^{n-1}$ uniformly at random. Assume that
\begin{equation}\equationlabel{mlarge}
n\ge c_0 k(r+k)\log(r+k)
\end{equation}
where $c_0$ is a sufficiently large constant.
Then, there exists a
set of orthonormal vectors $v_1,\dots,v_{r+k}\in\R^n$ such that
$\mathrm{span}\{v_1,\dots,v_{r+k}\}=\mathrm{span}\{u_1,\dots,u_r,n_1,\dots,n_k\}$ and
furthermore, with probability $99/100,$
\[
\mu_0([v_1\mid \dots\mid v_{r+k}])
\le
2\mu_0([u_1\mid\dots\mid u_k])+ O\left( \frac{k\log n}{r}\right)
\]
\end{lemma}

\begin{lemma}[\cite{HardtR12}]
\lemmalabel{mu0-infty}
Let $U$ be an orthonormal $n\times r$ matrix. Suppose $w\in\mathrm{range}(U)$ and
$\|w\|=1.$ Then,
\[
\|w\|_\infty^2\le \frac rn\cdot\mu_0(U)\mper
\]
\end{lemma}

\begin{lemma}
\label{lem:deflationcoherence}
Let $A \in \mathbb{R}^{n\times n}$ be a matrix. Define a set of vectors $s_1,\ldots,s_k$ and matrices $A'_1,\ldots,A'_k$ as follows. Let $A'_0 = A$. For each $i$, $s_i = A_{i-1}'t_i + c_i n_i$, where $t_i \in \mathbb{R}^n$ is an arbitrary vector, $c_i$ is an arbitrary real number, and $n_i \in \mathbb{S}^{n-1}$ is selected uniformly at random. Let $A_i' = A_{i-1}' - d_i s_i s_i^T$, where $d_i$ is an arbitrary real number. Then for all $i$:
$$\mu(A'_i) \leq 2r \mu(A) + O\left(i\log n\right)$$
\end{lemma}
\begin{proof}
We write $A = \sum_{j=i}^r \sigma_j u_j v_j^T$. The proof will follow easily from \lemmaref{mu0-perturb} and the following claim.

\begin{claim}
For $i \in \{0,\ldots,k\}$, let $w_1,\ldots, w_{r+i}$ denote the left singular vectors of $A'_i$. Then, $w_1,\ldots,w_{r+i} \in \mathrm{span}(u_1,\ldots,u_r,n_1,\ldots,n_i)$.
\end{claim}
\begin{proof}
We prove this by induction. The claim is immediate when $i = 0$, which forms
the base case. For the inductive case, consider $A'_i = A_{i-1}' - d_i s_i
s_i^T$. Write the singular value decomposition of $A'_{i-1}$ as: $A'_{i-1} =
\sum_{j=1}^{r+i-1} \sigma_j' y_j z_j^T$, and write the singular value
decomposition of $A'_i$ as: $A'_i = \sum_{j=1}^{r+i}\lambda_j w_j x_j^T$.   We
can also write $A'_i = \sum_{j=1}^{r+i-1} \sigma_j' y_j z_j^T - d_is_is_i^T$.
Therefore, for all $j,$
\begin{align*}
\lambda_j w_j = A'_i x_j
= \sum_{\ell = 1}^{r+i-1} \left(\sigma'_\ell \langle x_j, z_\ell \rangle\right)y_\ell - \left(d_i\langle x_j, s_i \rangle\right) s_i
\end{align*}
Therefore, $w_j \in \mathrm{span}(y_1,\ldots,y_{r+i-1}, s_i)$. But $s_i = A_{i-1}'t_i + c_i n_i$, so $s_i \in \mathrm{span}(y_1,\ldots,y_{r+i-1},n_i)$, and by our inductive assumption, $y_1,\ldots,y_{r+i-1} \in \mathrm{span}(u_1,\ldots,u_r,n_1,\ldots,n_{i-1})$. Therefore, we can conclude that $w_j \in \mathrm{span}(u_1,\ldots,u_r,n_1,\ldots,n_{i})$ for all $j$.
\end{proof}
By \lemmaref{mu0-perturb}, we can conclude that for all $j$, $w_j \in \mathrm{span}(v'_1,\ldots,v'_{r+i})$ such that $v'_1,\ldots,v'_{r+i}$ are orthonormal with:
$$\mu_0(v_1'\mid\ldots\mid v_{r+i}') \leq  2\mu_0(A) +  O\left( \frac{i\log n}{r}\right) \leq 2\mu(A) + O\left( \frac{i\log n}{r}\right).$$
Therefore, we have:
\begin{align*}
\mu(A_i') = n\cdot \max_{j \in [r+i]}\|w_j\|_\infty^2
\leq r \mu_0(v_1'\mid\ldots\mid v_{r+i}')
\leq  2r\mu (A) + O\left(i\log n\right)
\end{align*}
where the first inequality follows from \lemmaref{mu0-infty}.
\end{proof}

We are now ready to state our results for obtaining good rank $k$ approximations in the spectral norm. First, we translate our bounds from Section \ref{sec:incoherence} into a statement about rank-1 matrix approximation.

\begin{theorem}
Let $\gamma,\beta>0.$ Let $A$ be a matrix satisfying $\sigma_c\le (1-\gamma/2)\sigma_1$ for some
$c\ge1.$ Put $T=4\log(\sigma_1(A)).$
Further assume $A$ satisfies
\begin{equation}
\equationlabel{lower}
\|A\|_2 = \frac{\Theta T c
\sqrt{\mu(A)\log(1/\delta)}\log(n)}{\epsilon\gamma\beta}\mper
\end{equation}
for some sufficiently large constant $\Theta>0.$
Then, with probability $7/10,$ on input of $A,$ $T,$ $(\epsilon,\delta)$ and
$C\ge9\mu(A)\log(n),$ the algorithm rank-k$(A,T,\epsilon,\delta,1)$ outputs a rank $1$ matrix $A_1$ such that:
\[
\|A - A_1\|_2 \leq \sigma_2(A) + \beta\sigma_1(A) \mper
\]
\end{theorem}
\begin{proof}
This follows directly from \theoremref{powering}, together with Corollary \ref{cor:modified}, and the observation that:
$$\Pr\Set{|\hat{\sigma_{1}} - \sigma_1| \geq c\cdot \beta\sigma_1}
= \exp\left(-\epsilon' \beta\sigma_1\right) = O\left(n^{-cT\sqrt{\mu(A)}/\gamma}\right) = o(1)$$
Therefore, with probability at least $7/10$, both of the hypotheses of Corollary \ref{cor:modified} are satisfied.
\end{proof}
Our rank $k$ approximation result follows similarly, but we lose a factor of $\sqrt{r}$, where $r$ is the rank of the initial matrix to be approximated, due to the potential degradation in matrix coherence during the deflation process. It is not clear whether this factor of $r$ is necessary, or whether it is an artifact of our analysis.

\begin{theorem}
Let $\gamma, \beta >0.$ Fix a Let $A$ be a rank $r$ matrix such that there exist indices $c_1,\ldots,c_k$ such that for each $i$: $\sigma_{c_i} \le (1-\gamma/2)(\sigma_i - (i-1)\beta\sigma_1)$. Put $T=4\log(\sigma_1(A)).$
Further assume $A$ satisfies
\begin{equation}
\sigma_i(A) \ge \frac{\Theta Tc_k
\sqrt{(r\mu(A) + \sqrt{k}\log n)\log(k/\delta)}\log(n)}{\epsilon\gamma\beta}\mper
\end{equation}
for each $i \in [k]$.
for some sufficiently large constant $\Theta>0.$
Then, with probability $7/10,$ on input of $A,$ $T,$ $(\epsilon,\delta)$ and
$C\ge9\mu(A)\log(n),$ the algorithm rank-k$(A,k,T,\epsilon,\delta)$ outputs a rank $k$ matrix $A_k$ such that:
\[
\|A - A_k\|_2 \leq \sigma_{k+1}(A) +k\beta \sigma_1(A) \mper
\]
\end{theorem}

\begin{proof}
This follows directly from \theoremref{powering}, together with Corollary \ref{cor:modified}, and our bound on the degradation of the coherence of $A$ under deflation, Lemma \ref{lem:deflationcoherence}
\end{proof}

\section{Lower Bound}
\sectionlabel{lowerbound}

In this section, we prove a lower bound showing that our dependence on the
coherence $\mu$ is tight. For every value of $\mu \in [2,n]$, there is a
family of $n\times n$ matrices such that no $\epsilon$-differentially private
algorithm $A$ can compute a vector $A(M) = v$ with the guarantee that $\|Av\|
\geq \sigma_1 - o\left(\frac{\sqrt{\mu}}{\epsilon}\right)$.

\begin{theorem}
For every value of $C \in [2,n]$, there is a family of $n\times n$ matrices $\mathcal{M}_C$ such that:
\begin{enumerate}
\item For every $M \in \mathcal{M}_C$, $\mu(M) = C$
\item For any $\delta = \Omega(1/n)$, no $(\epsilon,\delta)$-differentially private algorithm $A:\mathbb{R}^{n\times n}\rightarrow \mathbb{R}^n$ has the guarantee that for every $M \in \mathcal{M}_C$, with constant probability, $A(M) = v$ such that $\|M v\|/\|v\| \geq \sigma_1(M) - o\left(\frac{\sqrt{C}}{\epsilon}\right)$.
\end{enumerate}
\end{theorem}
\begin{remark}
Note that this theorem shows that our upper bound for computing rank $1$ approximations to incoherent matrices is tight along the entire curve of values $\mu$, up to logarithmic factors.
\end{remark}
\begin{proof}
For each $C \in [2,n]$, we define our family of matrices $\mathcal{M}_C$ as follows. Let $\mathcal{D} \subset \mathbb{R}^n$ be the set of boolean valued vectors with exactly $n/2$ non-zero entries: $\mathcal{D} = \{D \in \{0,1\}^n : \|D\|_0 = n/2\}$. We will intuitively think of each $D \in \mathcal{D}$ as a private bit-valued database, whose entries we are protecting with a guarantee of differential privacy. For each $D \in \mathcal{D}$, let $\bar{D} = D/\|D\|_2$ be the rescaling of $D$ to a unit vector. Note that $\bar{D} \in \{0,\sqrt{2}/\sqrt{n}\}^n$. Define $s(C) = n/C$, and $u \in \mathbb{R}^d$ to be the vector such that $u_i = 1/\sqrt{s(C)}$ for $i \in \{1,\ldots,s(C)\}$ and $u_i = 0$ for $i > s(C)$. Finally, we define our class of matrices $\mathcal{M}_C$ to be:
\[
\mathcal{M}_C = \Set{M_D : M_D = \left(\sqrt{n s(C)/2}\right)u \cdot \bar{D}^T : D \in \mathcal{D}}
\]
Note that each $M_D \in \mathcal{M}$ is a matrix in which the first $s(C)$
rows are identical copies of the database $D \in \{0,1\}^n$, and the remaining
$n - s(C)$ rows are the zero vector. Moreover, for each $M \in \mathcal{M}_C$,
we have $$\sigma_1(M) = \left(\frac{\sqrt{n s(C)}}{\sqrt{2}}\right) =
\left(\frac{n}{\sqrt{2C}}\right),\quad\textrm{ and }\quad\mu(M) =
n\cdot\max\left(\frac{1}{s(C)},\frac{2}{n}\right) = \frac{n}{s(C)} = C$$ Now
consider any unit vector $v$. For each $M_D \in \mathcal{M}_C$, we have:
\begin{equation}
\label{eq:lowerboundapprox}
\|M_D v\|_2 = \sqrt{\frac{\sigma_1(M_D)^2\cdot \langle \bar{D}, v \rangle^2}{s(C)}\cdot s(C)} = \sigma_1(M_D)\cdot \langle \bar{D}, v\rangle
\end{equation}
Therefore, any unit vector $v$ such that $\|M_Dv\|_2 \geq \frac{999}{1000}\sigma_1(M_D)$ must be such that $\langle \bar{D}, v\rangle \geq \frac{999}{1000}$
However, if we view $D$ as being a private database, it is clear that it is not possible to privately approximate it well:

\begin{lemma}
\label{lem:reconstruction}
 For $\delta \leq 1/5$, Let $B:\mathbb{R}^n\rightarrow \mathbb{R}$ be a $(1,\delta)$-differentially private algorithm with respect to the entries of its input vector. Let $D \in \mathcal{D}$ be chosen uniformly at random. Then with probability $\geq 1/2$ $B(D) = v$ such that $\langle \bar{D},v/\|v\|_2\rangle \leq 1-\frac{2}{1000}$.
\end{lemma}
\begin{proof}
Let $D \in \mathcal{D}$ be a randomly chosen database $D \in \mathbb{R}^n$ with $\|D\|_0 = n/2$ entries. Let $\bar{D} = D/\|D\|$ be its normalization to a unit vector. Suppose that $v \in \mathbb{R}^n$ is a unit vector such that $\langle \bar{D}, v\rangle \geq 1-\alpha$. We may therefore write:
$$v = (1-\alpha)\bar{D} + \sqrt{1-(1-\alpha)^2}\bar{D}^\bot$$
where $\bar{D}^\bot$ is some unit vector orthogonal to $\bar{D}$. We therefore have:
\begin{eqnarray*}
\|\sqrt{\frac{n}{2}} v - D\|_1 &=& \|\alpha D + \sqrt{\frac{n}{2}}\sqrt{1-(1-\alpha)^2}\bar{D}^\bot\|_1 \\
&\leq& \alpha\|D\|_1 + \sqrt{\frac{n}{2}}\sqrt{1-(1-\alpha)^2}\|\bar{D}^\bot\|_1 \\
&\leq& \alpha\|D\|_1 + \sqrt{\frac{n}{2}}\sqrt{2\alpha}\|\bar{D}^\bot\|_1 \\
&\leq& \alpha\left(\frac{n}{2}\right) + \sqrt{2\alpha} \left(\frac{n}{\sqrt{2}}\right) \\
&\leq& \frac{n}{2} (3\sqrt{\alpha})
\end{eqnarray*}

Let $D^*$ denote the vector that results from setting $D^*_i = 1$ in each coordinate $i$ in which $\sqrt{\frac{n}{2}} v_i \geq 1/2$, and setting $D^*_i = 0$ in all other coordinates. Since $\|\sqrt{\frac{n}{2}} v - D\|_1 \leq \frac{n}{2}(3\sqrt{\alpha})$, it follows that $\|D^* - D\|_0 \leq \frac{n}{2} (6\sqrt{\alpha})$. Now consider the probability that a randomly chosen index $i \in \{i : D_i > 0\}$ is such that $D^*_i > 0$. This occurs with probability at least $1- (6\sqrt{\alpha})$. On the other hand, consider the probability that a randomly chosen index $j^* \in \{i : D_i = 0\}$ is such that $D^*_{j} \geq 0$. This occurs with probability at most $(6\sqrt{\alpha})$ Note also that because (over the random choice of $D$), each index $i \in D$ is set to $1$ uniformly at random, $i$ and $j$ are drawn from the same marginal distribution. Finally, consider the neighboring database $D' = D - \{i\} + \{j\}$, and note that $D'$ is also uniformly distributed among the set of databases $\mathcal{D}$. We therefore have that by differential privacy: $(1- (6\sqrt{\alpha})) \leq e\cdot(6\sqrt{\alpha}) + \delta$. If $\alpha < \frac{1}{1000}$ and $\delta < 1/5$, this is a contradiction.

\end{proof}
It remains to observe that changing a single entry of $D$ results in changing $s(C) = n/C$ entries of $M_D$. Therefore, by the composition properties of differential privacy, any algorithm $A:\mathbb{R}^{n\times n}$ which is $(\epsilon,\delta)$-differentially private with respect to entry changes in its input is $((n/C)\epsilon, (n/C)\delta)$-differentially private with respect to entry changes in $D$ when given $M_D$ as input. Therefore, Lemma \ref{lem:reconstruction} taken together with equation \ref{eq:lowerboundapprox} implies that if $\epsilon \leq C/n$ and $\delta \leq (C/(5n))$, then no $(\epsilon,\delta)$-differentially private algorithm, when given as input a uniformly randomly chosen matrix $M_D \in \mathcal{M}_C$ can with probability greater than $1/2$ return a vector $A(M_D) = v$ such that
$$\|M_D v\|_2 \geq \sigma_1(1-1/1000) = \sigma_1 - \left(\frac{n}{1000\sqrt{2C}}\right)$$

Finally, for point of contradiction, suppose that there was an $\epsilon$-differentially private algorithm that for every matrix $M$, with probability greater than $1/2$ returned a vector $A(M) = v$ such that $\|Mv\| \geq \sigma_1(M) - o\left(\frac{\sqrt{C}}{\epsilon}\right)$. Letting $\epsilon = C/n$, and letting $M = M_D$ be chosen from $\mathcal{M}_C$, we would have that:
$$\|Mv\|_2 \geq \sigma_1(M)- o\left(\frac{n}{\sqrt{C}}\right)$$
which is a direct contradiction. This completes the proof.

\end{proof}

\section{Conclusions and Open Problems}
We have shown nearly optimal \emph{data dependent} bounds for privately computing the top singular vector of a matrix, in terms of its $\mu$-coherence. We conclude with several specific open problems, as well as a general research agenda.

Specifically, it would be nice to resolve the following technical questions:
\begin{enumerate}
\item We have shown that our dependence on $\mu$-coherence is tight, but it remains possible that there might be a weaker notion of coherence that this or other algorithms could take advantage of. One candidate is $\mu_k$-coherence, which only bounds the magnitude of the entries in the top $k$ singular vectors, and leaves the others unconstrained. We do not know how to show that $\mu_k$-coherence is sufficient to bound the quality of the approximation to the top singular vector. However, as evidence of this conjecture, in Appendix \ref{section:localsensitivity}, we show that the \emph{local sensitivity} of the powering operation can be bounded in terms of $\mu_k$ coherence.
\item When we ``deflate'' $A$ so as to recurse and compute an approximation to the higher singular vectors, we lose a 1-time factor of $\sqrt{r}$ in the coherence, where $r$ is the rank of the matrix. As a result, our bounds for rank $k$ approximation for $k \geq 2$ have a dependence on the matrix rank. Can this factor of $\sqrt{r}$ be removed, or is it inherent?
\end{enumerate}

More generally, this paper is an instance of a broader research agenda: overcoming worst-case lower bounds in differential privacy by giving data-dependent accuracy bounds. In many settings (especially if the data set is small), the worst case bounds necessary to achieve differential privacy can be prohibitive. However, natural data sets tend to have structural properties (like low coherence) that can potentially be taken advantage of in a variety of settings. It would be interesting to understand the relevant features of the data that allow more accurate private analyses in domains other than spectral analysis.

\appendix

\section{Deviation bounds}

\begin{lemma}[Gaussian Anti-Concentration]
\lemmalabel{anti}
Let $\xi_i\sim N(0,1)$ and let $a_i\ge0$ for $1\le i\le n.$
Then, for every $\gamma>0,$
\[
\Pr\Set{\sum_{i=1}^na_i\xi_i^2 \le \gamma\sum_{i=1}^na_i} \le \sqrt{e\gamma}\mper
\]
\end{lemma}
We thank George Lowther for pointing out the following proof.
\begin{proof}
We may assume without loss of generality that $\sum_{i=1}^na_i=1$ and $\gamma<1.$
For every $\lambda>0,$ we have
\begin{align*}
\Pr\Set{\sum_{i=1}^na_i\xi_i^2\le \gamma\sum_{i=1}^na_i}
\le \E e^{\lambda\left(\gamma - \sum_{i=1}^na_i\xi_i^2\right)}
= e^{\lambda\gamma}\prod_{i=1}^n\E e^{a_i\xi_i^2}
&= e^{\lambda\gamma}\prod_{i=1}^n(1+2\lambda a_i)^{-1/2}\\
&\le e^{\lambda\gamma}(1+2\lambda)^{-1/2}\mper
\end{align*}
The claim follows by setting $\lambda = (\gamma^{-1}-1)/2.$
\end{proof}

The following direct consequence was needed earlier.

\begin{lemma}\lemmalabel{gaussian-norm}
Let $U$ be a $k$-dimensional subspace of $\R^n$ and let $g\sim N(0,1)^n.$
Then,
\begin{enumerate}
\item for every $\gamma>0,$
$\Pr\Set{\|P_Ug\|\le\sqrt{\gamma k}}\le \sqrt{e\gamma}\mper$
\item for every $t\ge1,$ we have
$\Pr\Set{\|P_Ug\|>\sqrt{t k}}\le \exp(-t)\mper$
\end{enumerate}
\end{lemma}
\begin{proof}
The first claim follows directly by \lemmaref{anti}. The second can be
verified by direct computation.
\end{proof}

\begin{theorem}[Chernoff bound]
\theoremlabel{chernoff}
Let the random variables $X_1,\dots,X_m$ be independent random variables.
Let $X = \sum_{i=1}^m X_i$ and let $\sigma^2=\Var X.$ Then, for any $t>0,$
\[
\Pr\Set{ \left|X-\E X\right| > t} \le
\exp\left(-\frac{t^2}{4\sigma^2}\right)\mper
\]
\end{theorem}

\section{Proofs from Section \ref{sec:deflation}}
\label{sec:deflationappendix}

\begin{lemma}[Deflation Lemma \cite{KapralovT13}]
Let $A$ be a symmetric matrix with eigenvalues $\lambda_1 \geq \ldots \geq \lambda_n$. There exists a constant $C > 0$ so that the following holds. Let $x$ be any unit vector such that $x^TAx \geq (1-\alpha/C)\lambda_1$, where $\alpha \in (0,1)$. Let $A' = A - t v\cdot v^T$, where $t \in (1\pm \alpha/C)x^TAx$. Denote the eigenvalues of $A'$ by $\lambda_1' \geq \ldots \geq \lambda_n'$.
\begin{enumerate}
\item $\lambda_k \leq \lambda'_{k-1} \leq \min(\lambda_{k-1}, \lambda_k+\alpha\lambda_1)$ for each $k \in \{1,\ldots,n\}$.
\end{enumerate}
\end{lemma}
Because our algorithm returns a vector $v$ with a guarantee on the quantity $\|A v\|$, rather than on the Rayleigh Quotient $v^T A v$, we must relate these two quantities, which we do presently.
\begin{lemma}
For any unit vector $x$, $|x^T A x| \leq \|Ax\|$
\end{lemma}
\begin{proof}
By the Cauchy-Schwarz inequality,
$|x^T A x| = |\langle x, Ax \rangle| \leq \|x\|\cdot \|Ax\| = \|A x\|.$
\end{proof}
We now prove a partial converse, for vectors $x$ such that $\|Ax\|$ is large.
\begin{lemma}
For any $0 \leq \alpha \leq 1/4$ and for any unit vector $x$ such that $\|Ax\| \geq (1-\alpha)\lambda_1$:
$$x^T A x \geq (1-5\alpha)\lambda_1$$
\end{lemma}
\begin{proof}
Let $v_1,\ldots,v_n$ denote the eigenvectors of $A$ in order from largest to
smallest eigenvalue: $|\lambda_1| \geq \ldots \geq |\lambda_n|$. Then we may
write $x = \sum_{i=1}^n \alpha_i\cdot v_i.$ Likewise
$Ax = \sum_{i = 1}^n \alpha_i\lambda_i \cdot v_i,$
where $\sum_{i=1}^n \alpha_i^2 = 1$, since $x$ is a unit vector.
Hence,
$\|A x\| = \sqrt{\sum_{i=1}^n \alpha_i^2\lambda_i^2}$ and
$x^TAx = \langle x, Ax \rangle = \sum_{i=1}^n \alpha_i^2\lambda_i.$

We define:
$$i^* \triangleq \max \{1 \leq i \leq n : |\lambda_i| \geq \lambda_1(1-4\alpha) \}$$
to be the largest index such that the $i^*$'th eigenvalue has magnitude at least $(1-4\alpha)\lambda_1$. Now define the quantities:
$$S_1 \triangleq \sum_{i=1}^{i^*}\alpha_i^2\ \ \ \ \ S_2 \triangleq \sum_{i=i^*+1}^n \alpha_i^2$$
and note that $S_2 = 1-S_1$.
We can calculate:
\begin{eqnarray*}
(1-\alpha)\lambda_1 &\leq& \sqrt{\sum_{i=1}^n \alpha_i^2\lambda_i^2} \\
&\leq& \sqrt{S_1 \lambda_1^2 + S_2 (1-4\alpha)\lambda_1^2} \\
&\leq& \lambda_1\left(\sqrt{S_1} + \sqrt{1-S_1}\sqrt{1-4\alpha}\right) \\
&\leq& \lambda_1\left(\sqrt{S_1} +  \sqrt{1-S_1}(1-2\alpha)\right)
\end{eqnarray*}
Solving for $S_1$, we find:
$$S_1 \geq \frac{1-4 \alpha +8 \alpha ^2-10 \alpha ^3+6 \alpha ^4+(-1+\alpha ) (-1+2 \alpha ) \sqrt{1+\alpha  (-2+3 \alpha )}}{2 (1+2 (-1+\alpha ) \alpha )^2} \geq 1 - 4\alpha^2.$$
Therefore, we also have $S_2 \leq 4\alpha^2$.
Finally, we may calculate:
\begin{align*}
x^TAx = \sum_{i=1}^n \alpha_i^2\lambda_i
\geq S_1(1-4\alpha)\lambda_1 - S_2(1-4\alpha)\lambda_1
\geq (1-8\alpha^2)(1-4\alpha)\lambda_1
\geq (1-5\alpha)\lambda_1
\end{align*}
where the last inequality holds since $\alpha \leq 1/4$
\end{proof}
As a corollary, we get a modified version of the deflation lemma of
\cite{KapralovT13}

\begin{corollary}[Modified Deflation Lemma \cite{KapralovT13}]
\label{cor:modified}
Let $A$ be a symmetric matrix with eigenvalues $\lambda_1 \geq \ldots \geq \lambda_n$. There exists a constant $C > 0$ so that the following holds. Let $x$ be any unit vector such that $\|Ax\| \geq (1-\alpha/C)\lambda_1$, where $\alpha \in (0,1)$. Let $A' = A - t v\cdot v^T$, where $t \in (1\pm \alpha/C)\|A x\|$. Denote the eigenvalues of $A'$ by $\lambda_1' \geq \ldots \geq \lambda_n'$.
\begin{enumerate}
\item $\lambda_k \leq \lambda'_{k-1} \leq \min(\lambda_{k-1}, \lambda_k+\alpha\lambda_1)$ for each $k \in \{1,\ldots,n\}$.
\end{enumerate}
\end{corollary}

\section{Perturbation bounds for matrix powers}
\label{section:localsensitivity}
In this section we prove a perturbation bound for matrix powers. The result
be seen as bounding the so-called \emph{local sensitivity} of power iteration.
Notably, we can use following notion weaker form of $\mu$-coherence that
depends only on the top few singular singular vectors.
\begin{definition}
Let $M$ be a real valued $m\times n$ matrix with singular value decomposition
$M=\sum_{i=1}^n \sigma_i u_i v_i^T,$ where
$\sigma_1\ge\sigma_2\ge...\ge\sigma_n\ge0.$ We define the \emph{top-$k$
coherence} of $M$ as
\[
\mu_k(M) = \max_{i=1}^k\max\Set{m\|u_i\|_\infty^2,
n\|v_i\|_\infty^2}\mper
\]
Note that $1\le\mu_k(M)\le \max\{m,n\}.$
\end{definition}

\begin{theorem}\theoremlabel{perturbation}
Let $q\ge 1$ be a number.
Let $M$ be a real valued $n\times n$ matrix with singular value decomposition
$M=\sum_{i=1}^n \sigma_i u_i v_i^T,$ where
$\sigma_1\ge\sigma_2\ge...\ge\sigma_n\ge0.$
Assume that $\sigma_1\ge 4q,$ $\sigma_{k+1}\le\sigma_1/2$ and $q\ge \log n + 1.$
Then, with $g\sim N(0,1)^n,$
\[
\E\| \left((M+E)^q - M^q\right)g \|_2
\le
9\min\Set{1,\sqrt{\frac {k\cdot\mu_k(M)}n}} q\sigma_1^{q-1}\mper
\]
\end{theorem}
\begin{remark}
We note that the above bound could easily be turned into a high probability guarantee.
\end{remark}
The next lemma will be helpful in simplifying some expressions arising
in the proof of the theorem.
\begin{lemma} \lemmalabel{simplify}
Let $E=e_se_t^T$ and $\alpha\ge1,$ Then,
\begin{itemize}
\item
$E^2=E$ if $s=t$ and $E^2=0$ otherwise.
\item
$EM^\alpha E  = C_\alpha E,$
where
$C_\alpha = \sum_{i=1}^n\sigma_i^\alpha\langle u_i,e_s\rangle\langle
v_i,e_t\rangle\mper$
\end{itemize}
\end{lemma}
\begin{proof}
Both claims are immediate.
\end{proof}
In the following we will use $C=\mu_k(M)$ as a shorthand.
\begin{lemma}\lemmalabel{Calpha}
Let $\delta=\min\Set{\sqrt{kC/n},1}\mper$ Then,
$C_\alpha \le \delta^2\sigma_1^\alpha +
\frac{\sigma_1^\alpha}{2^\alpha}\mper$
\end{lemma}
\begin{proof}
Appealing to the definition of $C_\alpha$ from \lemmaref{simplify}, we have
\begin{align*}
C_\alpha  =
\sum_{i=1}^n\sigma_i^\alpha\langle u_i,e_s\rangle\langle u_i,e_t\rangle
& = \sum_{i=1}^k\sigma_i^\alpha\langle u_i,e_s\rangle\langle v_i,e_t\rangle
+ \sum_{i=k+1}^n\sigma_i^\alpha\langle u_i,e_s\rangle\langle v_i,e_t\rangle \\
& \le \delta^2\sigma_1^\alpha + \frac{\sigma_1^\alpha}{2^\alpha}
\sum_{i=k+1}^n\langle u_i,e_s\rangle\langle v_i,e_t\rangle \\
& \le \delta^2\sigma_1^\alpha + \frac{\sigma_1^\alpha}{2^\alpha}\|e_s\|\|e_t\|
= \delta^2\sigma_1^\alpha + \frac{\sigma_1^\alpha}{2^\alpha}\mper
\end{align*}
\end{proof}
\begin{lemma}\lemmalabel{Malpha}
Recall that $\delta=\min\Set{\sqrt{kC/n},1}\mper$ We have,
\[
\E\|M^\alpha EM^\beta g\|
\le
\left(\delta\sigma_1^\alpha
+\frac{\sigma_1^\alpha}{2^\alpha}\right)
\left(\delta\sigma_1^\beta
+\frac{\sigma_1^\beta}{2^\beta}\right)
\mper
\]
\end{lemma}
\begin{proof}
First note that
\[
\E\|M^\alpha EM^\beta g\|^2
=\|M^\alpha e_s\|^2 \E \langle e_t^TM^\beta, g\rangle^2
=\|M^\alpha e_s\|^2 \|e_t^T M^\beta\|^2\mcom
\]
where we used that $g\sim N(0,1)^n.$ Hence, by Jensen's inequality,
\[
\E\|M^\alpha EM^\beta g\|
\le \sqrt{\E\|M^\alpha EM^\beta g\|^2}
= \|M^\alpha e_s\|\cdot\|e_t^T M^\beta\|\mper
\]
It remains to bound the right hand side of the previous inequality. Indeed,
\[
\|M^\alpha e_s\|
=\sqrt{\sum_{i=1}^n\sigma_i^{2\alpha}\langle u_i,e_s\rangle^2}
\le \sqrt{\frac {Ck}n}\sigma_1^\alpha
+\frac{\sigma_1^\alpha}{2^\alpha}\sqrt{\sum_{i=1}^k\langle u_i, e_s\rangle^2}
\le \sqrt{\frac {Ck}n}\sigma_1^\alpha
+\frac{\sigma_1^\alpha}{2^\alpha}\mper
\]
We can bound $\|e_t^T M^\beta\|$ with the same reasoning.
\end{proof}

\begin{lemma}\lemmalabel{Ag}
Let $A=M^{\alpha_1}EM^{\alpha_2}E\cdots EM^{\alpha_{\ell-1}} EM^{\alpha_\ell}$ with
$\alpha_i\ge0.$ Then,
\[
\E\|Ag\|
\le \left(\frac{\sigma_1}2\right)^{\sum_{i=1}^\ell\alpha_i}
+ \delta(1+\delta)^\ell\sigma_1^{\sum_{i=1}^\ell\alpha_i}\mper
\]
\end{lemma}
\begin{proof}
First we apply \lemmaref{simplify} to all intermediate
terms $EM^{\alpha_i}E$ where $i\in\{2,\dots,\ell-1\}.$ Then we apply
\lemmaref{Calpha} and \lemmaref{Malpha} to the remaining term.
Noting that $\delta^2\le \delta,$ since $0\le \delta\le 1,$ we have
established that
\[
\E\|Ag\|
\le \sigma_1^{\sum_{i=1}^\ell\alpha_i}\prod_{i=1}^\ell
\left(\delta+\frac1{2^{\alpha_i}}\right)
\]
On the other hand, it is not hard to see that the following inequality holds,
\[
\prod_{i=1}^\ell
\left(\delta+\frac1{2^{\alpha_i}}\right)
\le \left(\frac12\right)^{\sum_{i=1}^\ell\alpha_i}
+ \delta(1+\delta)^\ell \mper
\]
\end{proof}

We are now ready to prove \theoremref{perturbation}.

\begin{proof}[Proof of \theoremref{perturbation}]
Observe that the matrix $(M+E)^q-M^q$ equals the sum of $2^q-1$ matrices that
are either zero or of the form
$A=M^{\alpha_1}EM^{\alpha_2}E\cdots EM^{\alpha_{\ell-1}}
EM^{\alpha_\ell}$ as described in \lemmaref{Ag}. Let us say that
$\sum_{i=1}^\ell\alpha_i$ is the \emph{``order''} of the matrix $A.$ Clearly,
the order of the matrix $A$ is at most $q-\ell+1.$
Furthermore, there are at most $\binom qz$ matrices of order $z.$

Using the fact that $\binom qz \le q^{q-z}$ and the assumption that
$q\le \sigma_1/4,$ we will apply \lemmaref{Ag} to each such matrix and sum
over the resulting error terms:
\begin{align*}
\sum_{z=0}^{q-1}\binom qz \sigma_1^z\left(\frac{1}{2^{z}}
+ \delta(1+\delta)^{q-i+z}\right)
& \le
q\sum_{z=0}^{q-1}
\left(\frac{\sigma_1}{4}\right)^{q-z-1}\sigma_1^z\left(\frac{1}{2^{z}}
+ \delta 2^{q-z+1}\right)\\
& \le
q\sigma_1^{q-1}\sum_{z=0}^{q-1}\left(\frac1{2^{q-1}}+\frac{4\delta}{2^{q-z-1}}\right)\\
& \le q\sigma_1^{q-1}\left(8\delta + \frac q{2^{q-1}}\right)\\
& \le 9\delta q \sigma_1^{q-1}\mper
\end{align*}
In the last step we used that $\delta\ge\sqrt{1/n}$ and the assumption that
$q\ge \log(n)+1.$

The theorem follows now straightforwardly. By the previous argument and
linearity of expectation, we have
\[
\|((M+E)^q-M^q)g\|
\le 9\delta q \sigma_1^{q-1}
= 9\min\Set{1,\sqrt{\frac {Ck}n}} q\sigma_1^{q-1}\mper
\]
\end{proof}

\end{document}